\documentclass[sigconf]{aamas}  

\AtBeginDocument{%
  \providecommand\BibTeX{{%
    \normalfont B\kern-0.5em{\scshape i\kern-0.25em b}\kern-0.8em\TeX}}}
    
\usepackage{booktabs}    
\usepackage{flushend} 
\setcopyright{ifaamas}  
\copyrightyear{2020} 
\acmYear{2020} 
\acmDOI{} 
\acmPrice{} 
\acmISBN{} 
\acmConference[AAMAS'20]{Proc.\@ of the 19th International Conference on Autonomous Agents and Multiagent Systems (AAMAS 2020)}{May 9--13, 2020}{Auckland, New Zealand}{B.~An, N.~Yorke-Smith, A.~El~Fallah~Seghrouchni, G.~Sukthankar (eds.)}  

\usepackage{color}
\usepackage{courier}
\usepackage{tikz}
\usetikzlibrary{calc}
\usetikzlibrary{backgrounds,positioning}
\usetikzlibrary{decorations.pathreplacing}
\usepackage{diagbox}

\usepackage{xcolor}
\usepackage[linesnumbered,ruled,vlined,procnumbered]{algorithm2e}
\usepackage{etoolbox}

\makeatletter
\AtBeginEnvironment{procedure}{\let\c@algocf\c@procedure}
\makeatother
\usepackage{amsmath}
\DeclareMathOperator*{\argmax}{argmax}
\DeclareMathOperator*{\argmin}{argmin}

\SetCommentSty{mycommfont}

\SetKwInput{KwInput}{Input}                
\SetKwInput{KwOutput}{Output}

\tikzset{
	rednode/.style = {rectangle, draw=black!60, fill=red!20, very thick, minimum size=5mm,outer sep=0pt},
	bluenode/.style = {rectangle, draw=black!60, fill=blue!20, very thick, minimum size=5mm,outer sep=0pt},
	greennode/.style = {rectangle, draw=black!60, fill=green!20, very thick, minimum size=5mm,outer sep=0pt},
	whitenode/.style = {rectangle, draw=black!60, fill=white!20, very thick, minimum size=5mm,outer sep=0pt},
}
\begin{document}

\title{AED: An Anytime Evolutionary DCOP Algorithm}  


\author{Saaduddin Mahmud}
\affiliation{%
  \institution{Department of Computer Science and Engineering, University of Dhaka} 
}
\email{saadmahmud14@gmail.com}
\author{Moumita Choudhury}
\affiliation{%
	\institution{Department of Computer Science and Engineering, University of Dhaka} 
}
\email{moumitach22@gmail.com}
\author{Md. Mosaddek Khan}
\affiliation{%
	\institution{Department of Computer Science and Engineering, University of Dhaka} 
}
\email{mosaddek@du.ac.bd}
\author{Long Tran-Thanh}
\affiliation{%
	\institution{School of Electronics and Computer Science, University of Southampton} 
}
\email{ltt08r@ecs.soton.ac.uk}
\author{Nicholas R. Jennings}
\affiliation{%
	\institution{Departments of Computing and Electrical and Electronic Engineering, Imperial College London} 
}
\email{n.jennings@imperial.ac.uk}

\renewcommand{\shortauthors}{S. Mahmud et al.}

\begin{abstract}
  Evolutionary optimization is a generic population-based metaheuristic that can be adapted to solve a wide variety of optimization problems and has proven very effective for combinatorial optimization problems. However, the potential of this metaheuristic has not been utilized in Distributed Constraint Optimization Problems (DCOPs), a well-known class of combinatorial optimization problems prevalent in Multi-Agent Systems. In this paper, we present a novel population-based algorithm, Anytime Evolutionary DCOP (AED), that uses evolutionary optimization to solve DCOPs. In AED, the agents cooperatively construct an initial set of random solutions and gradually improve them through a new mechanism that considers an optimistic approximation of local benefits. Moreover, we present a new anytime update mechanism for AED that identifies the best among a distributed set of candidate solutions and notifies all the agents when a new best is found. In our theoretical analysis, we prove that AED is anytime. Finally, we present empirical results indicating AED outperforms the state-of-the-art DCOP algorithms in terms of solution quality.
\end{abstract}

%


\maketitle


\section{Introduction}
Distributed Constraint Optimization Problems (DCOPs) are a widely used framework to model constraint handling problems in cooperative Multi-Agent Systems (MAS). In particular, agents in this framework need to coordinate value assignments to their variables in such a way that minimizes constraint violations by optimizing their aggregated costs \cite{yokoo1998distributed}. This framework has been applied to various areas of multi-agent coordination, including distributed meeting scheduling \cite{Maheswaran2004TakingDT}, sensor networks \cite{farinelli2014agent}\cite{Choudhury2019APS} and smart grids \cite{fioretto2017distributed}.\par
Over the last two decades, several algorithms have been proposed to solve DCOPs, and they can be broadly classified into two classes: exact and non-exact. The former always provide an optimal solution of a given DCOP. Among the exact algorithms, SyncBB \cite{Hirayama1997DistributedPC}, ADOPT \cite{modi2005adopt}, DPOP \cite{Petcu2005ASM}, AFB \cite{AMIRGER2010AsynchronousFB}, BnB-ADOPT \cite{Yeoh2008BnBADOPTAA}, and PT-FB \cite{litov2017forward} are widely used. Since solving DCOPs optimally is NP-hard, scalability becomes an issue as the system grows. In contrast, non-exact algorithms compromise some solution quality for scalability. As a consequence, diverse classes of non-exact algorithms have been developed to deal with large-scale DCOPs. Among them, local search based algorithms are generally most inexpensive in terms of computational and communication cost. Some well-known algorithms of this class are DSA \cite{zhang2005distributed}, MGM \& MGM2 \cite{Maheswaran2004Distributed}, and GDBA \cite{okamoto2016distributed}. Also, in order to further enhance solution quality and incorporate an anytime property in local search based algorithms, the Anytime Local Search (ALS) framework \cite{zivan2014explorative} was introduced. While inference based non-exact approaches such as Max-Sum \cite{farinelli2008decentralised}\cite{speed}\cite{gene} and Max-Sum\_ADVP \cite{zivan2012max} have also gained attention due to their ability to handle n-ary constraints explicitly and guarantee optimality on acyclic constraint graphical representations of DCOPs. The third class of non-exact approaches that have been developed are sample-based algorithms (e.g. DUCT \cite{Ottens2012DUCTAU} and PD-Gibbs \cite{dgibbs}) in which the cooperative agents sample the search space in a decentralized manner to solve DCOPs.\par
More recently, a new class of non-exact DCOP algorithms have emerged in the literature through the introduction of a population-based algorithm ACO\_DCOP \cite{Chen2018AnAA}. ACO\_DCOP is derived from a centralized population-based approach called Ant Colony Optimization (ACO) \cite{Dorigo2006AntCO}. It has been empirically shown that ACO\_DCOP produces solutions with better quality than the aforementioned classes of non-exact DCOP algorithms \cite{Chen2018AnAA}. It is worth noting that although a wide variety of centralized population-based algorithms exist, ACO is the only such method that has been adapted to solve DCOPs. Among the remaining centralized population-based algorithms, a large portion is considered as evolutionary optimization techniques (e.g. Genetic Algorithm \cite{holland1992adaptation}, Evolutionary Programming \cite{ep}). Evolutionary optimization, as a population-based metaheuristic, has proven very effective in solving combinatorial optimization problems such as Traveling Salesman Problem \cite{Fogel1988AnEA}, Constraint Satisfaction Problem \cite{tsang1990applying}, and many others besides. However, no prior work exists that adapts evolutionary optimization techniques to solve DCOPs. Considering the effectiveness of evolutionary optimization techniques in solving combinatorial optimization problems along with the potential of population-based DCOP solver demonstrated by ACO\_DCOP motivates us to explore this nascent\:area.\par
Against this background, this paper proposes a novel population-based algorithm that uses evolutionary optimization to solve DCOPs. We call this Anytime Evolutionary DCOP (AED). In more detail, AED maintains a set of candidate solutions that are distributed among the agents, and they search for new improved solutions by modifying the candidate solutions. This modification is done through a new mechanism that considers an optimistic approximation of local benefits and utilizes the cooperative nature of the agents. Moreover, we introduce a new anytime update mechanism in order to identify the best among this distributed set of candidate solutions and help the agents to coordinate value assignments to their variables based on the best candidate solution. Our theoretical analysis proves that AED is anytime and empirical evaluation shows its superior solution quality compared to the state-of-the-art non-exact DCOP algorithms. \par
\section{Background}
In this section, we first describe DCOPs and Evolutionary Optimization in more detail. Then, we discuss challenges that need to be addressed in order to effectively extend evolutionary optimization in the context of DCOPs.
\subsection{Distributed Constraint Optimization Problems}
Formally, a DCOP is defined by a tuple $ \langle X,D,F,A,\delta \rangle $ \cite{modi2005adopt} where,
\begin{itemize}
    \item A is a set of agents $\{a_1, a_2, ..., a_n\}$.
    \item X is a set of discrete variables $\{x_1, x_2, ..., x_m\}$, which are being controlled by the set of agents A.
    \item D is a set of discrete and finite variable domains $\{D_1, D_2, ..., D_m\}$, where each $D_i$ is a set containing values which may be assigned to its associated variable $x_i$.
    \item F is a set of constraints $\{f_1,f_2,...,f_l\}$, where $f_i \in F$ is a function of a subset of variables $x^i \subseteq X$ defining the relationship among the variables in $x^i$. Thus, the function $f_i : \times_{x_j \in x^i} D_j \to \!R $ denotes the cost for each possible assignment of the variables in $x^i$. 
    \item $\delta: X \rightarrow
	A$ is a variable-to-agent mapping function \cite{nodeto} which assigns the control of each variable $x_i \in X$ to an agent of $A$. Each variable is controlled by a single agent. However, each agent can hold several variables. 
\end{itemize}
Within the framework, the objective of a DCOP algorithm is to produce $X^*$; a complete assignment that minimizes\footnote{For a maximization problem $\argmin$ is replaced with $\argmax$ in Equation~\ref{eqobj}.} the aggregated cost of the constraints as shown in Equation~\ref{eqobj}.
\vspace{-2mm}
\begin{equation}
    X^* = \argmin_X \sum_{i=1}^{l} f_i(x^i)
    \label{eqobj}
\end{equation}
For ease of understanding, we assume that each agent controls one variable. Thus, the terms `variable' and `agent' are used interchangeably throughout this paper. Figure~\ref{dcopex}a illustrates a sample DCOP using a constraint graph where each node represents an agent $a_i \in A$ labelled by a variable $x_i \in X$ that it controls and each edge represents a function $f_i \in F$ connecting all $x_j \in x^i$. Figure~\ref{dcopex}b shows the corresponding cost tables.
\begin{figure}[t]
\centering
\scalebox{.8}{
  \begin{tikzpicture}
        [
        roundnode/.style={circle, draw=green!60, fill=green!5, very thick, minimum size=7mm},
        ]
        \node[roundnode]    at(1,0)  (x1)                              {$x_1$};
        \node[roundnode]    at(0,-1)  (x2)                              {$x_2$};
        \node[roundnode]    at(2,-1)  (x3)                              {$x_3$};
        \node[roundnode]    at(0,-2)  (x4)                              {$x_4$};
         
        \draw (x1) -- (x2);
        \draw (x1) -- (x3);
        \draw (x3) -- (x2);
        \draw (x4) -- (x2);
        \node  at (1,-3)
        {
            (a) A constraint graph
        };
        \node at (4,0)  {
            \begin{tabular}{|l|c|c|}\hline
            \diagbox[width=2.5em]{$x_1$}{$x_2$}&
              1 & 2 \\ \hline
            1 & 7 & 12 \\ \hline
            2 & 3 & 15 \\ \hline
            \end{tabular}
        };
        \node at (6.55,0)  {
            \begin{tabular}{|l|c|c|}\hline
            \diagbox[width=2.5em]{$x_2$}{$x_3$}&
              1 & 2 \\ \hline
            1 & 2 & 7 \\ \hline
            2 & 11 & 18 \\ \hline
            \end{tabular}
        };
        \node at (4,-1.65)  {
            \begin{tabular}{|l|c|c|}\hline
            \diagbox[width=2.5em]{$x_2$}{$x_4$}&
              1 & 2 \\ \hline
            1 & 8 & 4 \\ \hline
            2 & 15 & 6 \\ \hline
            \end{tabular}
        };
        \node at (6.55,-1.65)  {
            \begin{tabular}{|l|c|c|}\hline
            \diagbox[width=2.5em]{$x_1$}{$x_3$}&
              1 & 2 \\ \hline
            1 & 9 & 13 \\ \hline
            2 & 12 & 5 \\ \hline
            \end{tabular}
        };
        \node  at (5.5,-3)
        {
            (b) Cost tables 
        };
    \end{tikzpicture}
}
\vspace{-2mm}
\caption{Example DCOP}
\label{dcopex}
\vspace{1mm}
\end{figure}
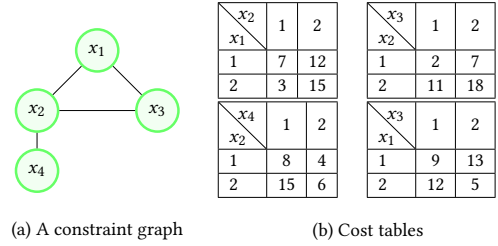
\subsection{Evolutionary Optimization}

Evolutionary optimization is a generic population-based metaheuristic inspired by biological evolutionary mechanisms such as Selection, Reproduction and Migration. The core mechanism of evolutionary optimization techniques can be summarized in three steps. In the first step, an initial population is generated randomly. A population is a set of `individuals', each of which is a candidate solution of the corresponding optimization problem. Besides, a fitness function is defined to evaluate the quality of an individual concerning a global objective. The fitness of all the individuals in the initial population is also calculated. In the second step, a subset of the population is selected based on their fitness to reproduce new individuals. This process is known as Selection. In the final step, new individuals are created using the selected subset of the population and their fitness is evaluated. New individuals then replace a subset of old individuals. Evolutionary optimization performs both the second and the third steps iteratively, which results in a gradual improvement in the quality of individuals. An additional step is performed at regular intervals by some parallel/distributed evolutionary optimization models that concurrently maintain multiple sub-populations instead of a single population. In this step, individuals are exchanged between sub-populations. This process is known as Migration, and this interval is known as the Migration\:Interval. 

\subsection{Challenges}
 We need to address the following challenges in order to develop an effective anytime algorithm that uses evolutionary optimization to solve DCOPs: 
\begin{itemize}
    \item \textbf{Individual and fitness:} We need to define an individual that represents a solution of a DCOP along with a fitness function to evaluate its quality concerning Equation \ref{eqobj}. We also need to provide a method for calculating this fitness function in a distributed manner.
    \item \textbf{Population:} We need to provide a strategy to maintain the population collectively among the agents. Although creating an initial random population is a trivial task for centralized problems, we need to find a distributed method to construct an initial random population for a DCOP. 
    \item \textbf{Reproduction mechanism:} In the DCOP framework, information related to the entire problem is not available to any single agent. So it is necessary to design a Reproduction method that can utilize information available to a single agent along with the cooperative nature of the agents. 
    \item \textbf{Anytime update mechanism:} We need to design an anytime update mechanism that can successfully perform the following tasks -- (i) Identify the best individual in a population that is distributed among the agents. (ii) Notify all the agents when a new best individual is found. (iii) Help coordinate the variable assignment decision based on the best individual in a population.     
\end{itemize}
In the following section, we describe our method that addresses the above challenges. 

\section{The AED Algorithm}
AED is a synchronous iterative algorithm that consists of two phases: Initialization and Optimization. During the former, agents initially order themselves into a pseudo-tree, then initialize the necessary variables and parameters. Finally, they make a random assignment to the variables they control and cooperatively construct the initial population. During the latter phase, agents iteratively improve this initial set of solutions using the cooperation of their neighbours. When an agent detects a better solution, it notifies other agents. Moreover, all the agents synchronously update their assignments based on the best of the individuals reported to them so far. This results in a monotonic improvement of the global objective.  Algorithm~\ref{algo:AED} shows the pseudo-code for AED. For ease of understanding, we show the process of initialization and anytime update separately in Procedure 1 and Procedure 2, respectively. Note that the initialization phase addresses the first two of our challenges, while the optimization phase addresses the rest.\par
\begin{figure}[t]
    \centering
  \scalebox{.8}{
  \begin{tikzpicture}
        [
        roundnode/.style={circle, draw=green!60, fill=green!5, very thick, minimum size=7mm},
        ]
        \node[roundnode]    at(1,0)  (x4)                              {$x_4$};
        \node[roundnode]    at(1,-1)  (x2)                              {$x_2$};
        \node[roundnode]    at(0,-2)  (x1)                              {$x_1$};
        \node[roundnode]    at(2,-2)  (x3)                              {$x_3$};
        \draw (x4) -- (x2);
        \draw (x2) -- (x3);
        \draw (x2) -- (x1);
        \draw[dotted] (x1) -- (x3);
\end{tikzpicture}}
\vspace{-2mm}
\caption{BFS Pseudo-Tree}
\label{algo:pt}
\end{figure}
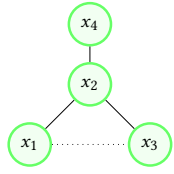
\textbf{The Initialization Phase} of AED  consists of two parts; pseudo-tree construction and running INIT (Procedure 1) that initializes the population, parameters and variables (Algorithm~\ref{algo:AED}: Line 1-2). This phase starts by ordering the agents into a Pseudo-Tree. This ordering serves two purposes. It helps in the construction of the initial population and facilitates ANYTIME-UPDATE (Procedure 2) during the optimization phase. Even though either of the BFS or DFS pseudo-tree can be used, AED uses BFS Pseudo-tree\footnote{We suggest using the algorithm described in \cite{hen2017improved}. Height can be easily calculated by utilizing LAYER information.}. This is because it generally produces a pseudo-tree with smaller height \cite{hen2017improved}, which improves the performance of ANYTIME-UPDATE (see Theoretical Analysis for details). Figure 2 shows an example of a BFS pseudo-tree constructed from the constraint graph shown in Figure 1a having $x_4$ as the root. Here, the height\footnote{Length of the longest path in the pseudo-tree.} (i.e. H = 2) of this pseudo-tree is calculated during the time of construction and is maintained by all agents. From this point, $N_i$ refers to the set of neighbours; $C_i \subseteq N_i$ refers to the set of child nodes and ${PR}_i$ refers to the parent of an agent $a_i$ in the pseudo-tree. For instance, we can see in Figure 2 that $N_2 = \{a_1,a_3,a_4\}$, $C_2 = \{a_1,a_3\}$ and ${PR}_2 = a_4$ for agent $a_2$. After the pseudo-tree construction, all the agents synchronously call the procedure INIT (Algorithm~\ref{algo:AED}: Line 2).\par 
\begin{algorithm}[t]
\DontPrintSemicolon
\small
   Construct pseudo-tree\;
   Every agent $a_i$ calls INIT( )\;
   \While{Stop condition not met each agent $a_i$}
   {
   		$P_{selected} \leftarrow Select_{rp}(|N_i|*ER)$ \;
   		$P_{new} \leftarrow$ Partition $P_{selected}$ into equal size subsets $\{P_{new}^{n_1},...,P_{new}^{n_{|N_i|}}\}$\;
   		\For{$n_j \in N_i$}{
   		    Modify individuals in $P_{new}^{n_j}$ by Equations~\ref{eqw1},~\ref{eqw2},~\ref{eqw3},~\ref{eqdel}\;
   		    Send message $P_{new}^{n_j}$ to $n_j$\;
   		}
   		\For{$P_{received}^{n_{j}}$ received from $n_j \in N_i$}{
   		    Modify individuals in $P_{received}^{n_{j}}$ by Equations ~\ref{eqg1},~\ref{eqdel}\;
   		    Send $P_{received}^{n_{j}}$ to $n_j$\;
   		}
   		\For{$n_j \in N_i$}{
   		    Receive $P_{new}^{n_j}$ back from $n_j$
   		}
   		$P_{a_i} \leftarrow P_{a_i} \cup P_{new}$\; 
   		$B \leftarrow \argmin_{I \in P_{a_i}} I.fitness$\;
   		$ANYTIME-UPDATE(B)$\;
   		$P_{a_i} \leftarrow Select_{wrp}(|N_i|*ER)$ \;
   		\uIf{$Itr$ = $Itr_M+MI$ }{
   		    \For{$n_j \in N_i$}{
                Send $Select_{wrp}(ER)$ to $n_j$\;
            }
            \For{$P_{received}^{n_{j}}$ received from $n_j \in N_i$}{
   		        $P_{a_i} \leftarrow P_{a_i} \cup P_{received}^{n_{j}}$\;
   		    }
            
          }
   }
\caption{Anytime Evolutionary DCOP}
\label{algo:AED}
\end{algorithm}
INIT starts by initializing all the parameters and variables to their default values\footnote{AED takes a default value for each of the parameters as input. Default values of the variables are discussed later in this section.}.  Then each agent $a_i$ sets its variable $x_i$ to a random value from its domain $D_i$. Lines 3 to 25 of Procedure 1 describe the initial population construction process. In AED, we define population \textit{P} as a set of individuals that are collectively maintained by all the agents and local population $P_{a_i} \subseteq P$ as the subset of the population maintained by agent $a_i$. An individual in AED is represented by a complete assignment of variables in X and fitness is calculated using a fitness function shown in Equation~\ref{eqfit}. This function calculates the aggregated cost of constraints yielded by the assignment. Hence, optimizing this fitness function results in an optimal solution for the corresponding DCOP. 
\begin{equation}
    fitness = \sum_{f_i\in F}^{} f_i(x^i) 
    \label{eqfit}
\end{equation}
Note that a single agent can not calculate the fitness function. Rather it is calculated in parts with the cooperation of all the agents during the construction process. Moreover, the fitness value is added to the representation of an individual because it enables an agent to recalculate the fitness when a new individual is constructed only using local information. We take I = $\{x_1=1, x_2=2, x_3=1, x_4=2,fitness = 38\}$ as an example of a complete individual from the DCOP shown in Figure 1. We use dot(.) notation to refer to a specific element of an individual. For example $I.x_1$ refers to $x_1$ in I.
Additionally, we define a Merger operation of two individuals under construction, $I_1, I_2$ as Merge($I_1, I_2$). This operation constructs a new individual $I_3$ by aggregating the assignments and setting $I_3.fitness = I_1.fitness+I_2.fitness$. We define an extended Merge operation for two ordered sets of individuals $S_1$ and $S_2$ as $Merge(S_1,S_2) = \{I_i: Merge(S_1.I_i,S_2.I_i)\}$ where $I_i$ is the i-th individual in a set. \par
 At the beginning of the construction process, each agent $a_i$ sets $P_{a_i}$ to a set of empty individuals\footnote{Individuals with no assignment and fitness set to 0.}. The size of the initial $P_{a_i}$ is defined by parameter IN. Then for each individual $I \in P_{a_i}$, agent $a_i$ makes a random assignment to $I.x_i$. After that each agent $a_i$ executes a merger operation on $P_{a_i}$ with each local population maintained by agents in $N_i$ (Procedure 1: Line 2-8). At this point, an individual $I \in P_{a_i}$ consists of an assignment of variables controlled by $a_i$, and agents in $N_i$ with fitness set to zero. For example, I = $\{x_1=1, x_2=2, x_3=1, fitness = 0\}$ represents an individual of $P_{a_3}$. The fitness of each individual is then set to the local cost according to their current assignment (Procedure 1: Line 9-10). Hence, the individual I from the previous example becomes $\{x_1=1, x_2=2, x_3=1, fitness = 20\}$. In the next step, each agent $a_i$ executes a merger operation on $P_{a_i}$ with each local population that is maintained by the agents in $C_i$. Then each agent $a_i$ sends $P_{a_i}$ to ${PR}_i$ apart from the root (Procedure 1: Line 11-18). At the end of this step, the local population maintained by the root consists of complete individuals. However, their fitness is twice its actual value since each constraint is calculated twice. Therefore, the root agent at this stage corrects all the fitness values (Procedure 1: Lines 20-21). Finally, the local population of the root agent is distributed through the network so that agents can initialize their local population (Procedure 1: Line 22-25). This concludes the initialization phase and after that, all the agents synchronously start the optimization phase in order to improve this initial population iteratively.\par   	
\begin{procedure}[t]
  \caption{INIT()( )}
  \DontPrintSemicolon
  \small
  Initialize algorithm parameters IN, ER, $R_{max}, \alpha, \beta$, MI and variables LB, GB, FM, UM\;
  $x_i \leftarrow$ random value from $D_i$\;
  $P_{a_i} \leftarrow $Set of empty individuals\;
  \For{Individual $I \in P_{a_i}$ }{
    $I.x_i \leftarrow $a random value from $D_i$\;
  }
  Send $P_{a_i}$ to agents in $N_i$  \;
  \For{$P_{n_j}$ received from $n_j \in N_i$}{
    $P_{a_i} \leftarrow Merge(P_{a_i},P_{n_j})$\;
  }
  \For{Individual $I \in P_{a_i}$ }{
    $I.fitness \leftarrow \sum_{n_j\in N_i}^{} Cost_{i,j}(I.x_i,I.x_j) $\;
  }
  \eIf{$|C| = 0$}
  {
      Send $P_{a_i}$ to ${PR}_i$ \;
  }
  {
        Wait until received $P_{c_j}$ from all $c_j \in C_i$\;
        \For{$P_{c_j}$ received from $c_j \in C_i$ }
        {
            $P_{a_i} \leftarrow Merge(P_{a_i},P_{c_j})$\;
        }
        \eIf{$a_i \not= root$}
        {
            Send $P_{a_i}$ to ${PR}_i$\;
        }
        {
            \For{Individual $I \in P_{a_i}$ }
            {
                $I.fitness \leftarrow I.fitness/2 $\;
            }   
            Send $P_{a_i}$ to all agent in $C_i$\;
      }  
  }
  \If{Received $P_{{PR}_i}$ from ${PR}_i$}{
    $P_{a_i} \leftarrow P_{{PR}_i}$\;
    Send $P_{a_i}$ to all agent in $C_i$\;
  }
\end{procedure}

\textbf{The Optimization Phase} of AED consists of five steps, namely Selection, Reproduction, ANYTIME-UPDATE, Reinsertion and Migration. An agent $a_i$ begins an iteration of this phase by selecting individuals from $P_{a_i}$ for the Reproduction step (Algorithm 1: Line 4). Prior to this selection, all the individuals are ranked from $(0, R_{max}]$ based on their relative fitness in the local population $P_{a_i}$. The rank $R_j$ of an individual $I_j \in P_{a_i}$ is calculated using Equation ~\ref{eqs1}. Here, $I_{best}$ and $I_{worst}$ are the individuals with the lowest and highest fitness in $P_{a_i}$ respectively\footnote{For minimization problems, a lower value of fitness is better.}. We define $Select_{rp}(S)$  as the process of taking a sample with replacement\footnote{Any individual can be selected more than once.} of size S from population $P_{a_i}$ based on the probability calculated using Equation ~\ref{eqs2}. As $\alpha$ increases in Equation ~\ref{eqs2}, the fitness vs. selection probability curve gets steeper. As a consequence, individuals with better fitness get selected more often. In this way, $\alpha$ controls the exploration and exploitation dynamics in the Selection mechanism (See Section 5 for more details). For example, assume $P_{a_i}$ consists of 3 individuals $I_1, I_2, I_3$ with fitness 16, 30, 40 respectively and $R_{max} = 5$. Then Equations ~\ref{eqs1} and ~\ref{eqs2} will yield, $P(I_1)=0.676, P(I_2)=0.297, P(I_3)=0.027$ if $\alpha = 1$ and $P(I_1)=0.92153, P(I_2)=0.07842, P(I_3)=0.00005$ if $\alpha = 3$. During this step, each agent $a_i$ selects $|N_i|*ER$ individuals from $P_{a_i}$ which we define as $P_{selected}$. 
\begin{equation}
    R_j =  R_{max}* \frac{|I_{worst}.fitness - I_j.fitness|+1}{|I_{worst}.fitness - I_{best}.fitness|+1}
    \label{eqs1}
\end{equation}
\begin{equation}
    P(I_j) = \frac{R_{j}^{\alpha}}{\sum_{I_k \in P_{a_i}} R_k^{\alpha}}
    \label{eqs2}
\end{equation}\par
Now, lines 5 to 11 of Algorithm 1 illustrate our proposed Reproduction mechanism. Agents start this step by partitioning $P_{selected}$ into $|N_i|$ subsets of size ER. Then each subset is randomly assigned to a unique neighbour. The subset assigned to $n_j \in N_i$ is denoted by $P_{new}^{n_j}$. An agent $a_i$ creates a new individual from each $I\in P_{new}^{n_j}$ with cooperation of neighbour $n_j$. Initially, agent $a_i$ changes assignment $I.{x_i}$ by sampling from its domain $D_i$ using Equations ~\ref{eqw1}, ~\ref{eqw2},~\ref{eqw3}. Then, $P_{new}^{n_j}$ is sent to $n_j$. Agent $n_j$ updates its assignment of $I.{x_j}$ for each $I\in P_{received}^{a_i}$ (i.e. $P_{new}^{n_j}$) using Equation ~\ref{eqg1}. Additionally, both agents $a_i$ and $n_j$ update the fitness of the individual $I$ by adding $\delta_i$ and $\delta_j$ to I.fitness, respectively. Here, $\delta_*$ is calculated using Equation ~\ref{eqdel} where $I.x_*^{new}$ and $I.x_*^{old}$ are the old and new values of $I.x_*$, respectively.\par
\begin{equation}
    O_{d_i} = \sum_{n_k \in N_i\setminus n_j}Cost_{i,k}(I.x_i,I.x_k) + \min_{d_j \in D_j} Cost_{i,j}(I.{x_i},d_j)
    \label{eqw1}
\end{equation}
\begin{equation}
    W_{d_i} = O_{max}* \frac{|O_{worst} - O_{d_i}|+1}{|O_{worst} - O_{best}|+1}
    \label{eqw2}
\end{equation}
\begin{equation}
    P(d_i) = \frac{W_{d_i}^{\beta}}{\sum_{d_k \in D_i}W_{d_k}^{\beta}}
    \label{eqw3}
\end{equation}
 \begin{equation}
    I.x_j = \argmin_{d_j \in D_j} \sum_{n_k \in N_j}Cost_{j,k}(d_j,I.x_k)
    \label{eqg1}
\end{equation}
 \begin{equation}
    \delta_{*} = \sum_{n_k \in N_*}Cost_{*,k}(I.x_{*}^{new},I.x_k)-Cost_{*,k}(I.x_*^{old},I.x_k)
    \label{eqdel}
\end{equation}
For example, agent $a_3$ of Figure 1 creates a new individual from $I = \{x_1=1, x_2=2, x_3=2, x_4=2,fitness = 49\}$ with the help of neighbour $a_2$. Here, the domain of agent $a_3$ and $a_2$ is $\{1,2\}$. Initially, agent $a_3$ calculates P(1) = 0.90 and P(2) = 0.10 using Equation ~\ref{eqw1},~\ref{eqw2},~\ref{eqw3} ($\beta = 1$). It then updates $I.x_3$ by sampling this probability distribution. The fitness is also updated by adding $\delta_i$ (= -11). Let the updated I be $\{x_1=1, x_2=2, x_3=1, x_4=2,fitness = 38\}$, it is then sent to $a_2$. Based on Equation ~\ref{eqg1}, the new value of $I.x_2$ should be 1. Now, agent $a_2$ updates $I.x_2$ along with the fitness by adding $\delta_j$ (= -16) and sends I back to $a_3$. Hence, Agent $a_3$ receives $I = \{x_1=1, x_2=1, x_3=1, x_4=2,fitness = 22\}$.\par    
To summarize the Reproduction mechanism, each agent $a_i$ picks a neighbour $n_j$ randomly for each $I\in P_{selected}$. Agent $a_i$ then updates $I.x_i$ by sampling based on the most optimistic cost (i.e. the lowest cost) of the constraint between $a_i$ and $n_j$ and aggregated cost of the remaining local constraints. This cost represents the optimistic local benefit for each domain value. Then $n_j$ sets $I.x_j$ to a value that complements the optimistic change in $I.x_i$ most. The key insight of this mechanism is that it not only takes into account the improvement in fitness that the change in $I.x_i$ will bring but also considers the potential improvement the change in $I.x_j$ will bring. Moreover, note that the parameter $\beta$ in Equation ~\ref{eqw3} plays a similar role as parameter $\alpha$ in Equation 3 (See Section 5 for details). After collecting the newly constructed individuals from neighbours they are added to $P_{a_i}$ (Algorithm 1: Line 12-14). Then the best individual B in $P_{a_i}$ is sent for ANYTIME-UPDATE (Algorithm 1: Line 15-16).\par 
To facilitate the anytime update mechanism, each agent maintains four variables LB, GB, FM, UM. LB (Local Best) and GB (Global Best) are initialized to empty individuals with fitness set to infinity. FM and UM are initialized to $\emptyset$. Additionally, GB is stored with a version tag and each agent maintains previous versions of GB having version tags in the range $[Itr-H+1, Itr]$ (see the Theoretical section for details). Here, Itr refers to the current iteration number. We use $GB^j$ to refer to the latest version of GB with version tag not exceeding j. Ours proposed anytime update mechanism works as follows. Each agent keeps track of two different best, LB and GB. Whenever the fitness of LB becomes less than GB, it has the potential to be the global best solution. So it gets reported to the root through the propagation of a Found message up the pseudo-tree. Since the root gets reports from all the agents, it can identify the true global best solution, and notify all the agents by propagating an Update message down to the pseudo tree. The root also adds the version tag in the Update message to help coordinate variable assignment. Now, ANYTIME-UPDATE starts by keeping LB updated with the best individual B in $P_{a_i}$. In line 3 of Procedure 2, agents try to identify whether LB is the potential global best. When identified and if the identifying agent is the root, it is the true global best and an Update message UM is constructed. If the agent is not the root, it is a potential global best and a Found message FM is constructed (Procedure 2: Lines 4-8). Each agent forwards the message UM to agents in $C_i$ and the message FM to the ${PR}_i$. Upon receiving these messages, an agent takes the following actions: 
\begin{itemize}
    \item If an Update message is received then an agent updates both its GB and LB. Additionally, the agent saves the Update message in UM and sends it to all the agents in $C_i$ during the next iteration (Procedure 2: Lines 12-15). 
    \item If a Found message is received and it is better than LB, only LB is updated. If this remains a potential global best it will be sent to ${PR}_i$ during next iteration (Procedure 2: Lines 16-17).
\end{itemize}\par
\begin{procedure}[t]
  \caption{ANYTIME-UPDATE(B)}
    \DontPrintSemicolon
    \small
   \If{$B.fitness < LB.fitness$}{
        $LB \leftarrow B$\;
   }
   \If{$LB.fitness < GB^{Itr}.fitness$}{
        \eIf{$a_i = root$}{
            $GB^{itr} \leftarrow LB$\;
            $UM \leftarrow \{Version:Itr,Individual:LB\}$
        }{
            $FM \leftarrow \{Individual:LB\}$
        }
   }
   Send Update Message UM to agents in $C_i$ and Found Message FM to ${PR}_i$ \;
   $FM \leftarrow \emptyset$\;
   $UM \leftarrow \emptyset$\;
   \If{Received update message M and $M \not= \emptyset $}{
        $GB^{M.Version} \leftarrow M.individual$\;
        $LB \leftarrow $Best between LB and M.individual\;
        $UM \leftarrow M$\;
   }
   \If{Received found message M and $M \not= \emptyset $ and $M.individual.fitness < LB.fitness$}{
        $LB \leftarrow M.individual$\;
   }
   \If{$Itr >= H$}{
            $x_i = GB^{Itr-H+1}.x_i$\;
   }
\end{procedure}
An agent $a_i$ then updates the assignment of $x_i$ using $GB^{Itr-H+1}$ (Procedure 2: Lines 18-19). Agents make decisions based on $GB^{Itr-H+1}$ instead of the potentially newer $GB^{Itr}$ so that decisions are made based on the same version of GB. $GB^{Itr-H+1}$ will be same for all agents since it takes at most H iterations for an Update message to propagate to all the agents. For example, assume agent $a_1$ from Figure 2 finds a potential best individual I at $Itr = 3$. Unless it gets replaced by a better individual, it will reach the root $a_4$ via agent $a_2$ through a Found message at $Itr =4$. Then $a_4$ constructs an Update message $\{Version: 5, Individual:I\}$ at $Itr=5$. This message will reach all the agents by $Itr=6$ and the agents save it as $GB^5=I$. Finally, at $Itr=6$ agents assign their variables using $GB^{6-2+1}=GB^{5}$ which is the best individual found at $Itr=3$.\par   
After ANYTIME-UPDATE each agent performs Reinsertion, i.e. updates its local population with new individuals. At first each agent adds newly constructed individuals $P_{new}$ to $P_{a_i}$ (Algorithm 1: line 14). After that, each agent $a_i$ updates their $P_{a_i}$ by keeping a sample of size $|N_i|*ER$ and discarding the rest based on their fitness (Algorithm 1: line 17). This sample is taken using $Select_{wrp}(S)$ which is the same as $Select_{rp}(S,)$ except agents sample without replacement\footnote{Each individual can be selected at most once.}. This sampling method keeps the local population $P_{a_i}$ diverse by selecting a unique set of individuals.\par 
Finally, Migration, an essential step of AED, takes place on every MI iteration. We sketch this in lines 18-22 of Algorithm 1. For this step, we define $Itr_{M}$ as the iteration number when the last Migration occurred. Migration is a simple process of exchanging individuals among the neighbours. In AED, the Reproduction mechanism utilizes local cooperation, so only a subset of variables of an individual change. However, because of Migration, different agents can change a different subset of variables as individuals get to traverse the network through this mechanism. Hence, this step plays an essential role in the optimization process of AED. During this step, an agent $a_i$ selects a sample of size ER using $Select_{wrp}(S)$ for each $n_j\in N_i$ and sends a copy of those individuals to that neighbour. Upon collecting individuals from all the neighbours, an agent, $a_i$ adds them to its local population $P_{a_i}$. This concludes an iteration of the optimization phase and every step repeats during the subsequent iterations..\par 

\section{Theoretical Analysis}
In this section, we first prove that AED is anytime, that is the quality of solutions found by AED increase monotonically. Then we analyze the complexity of AED in terms of communication, computation and memory requirements.
\begin{lemma} 
At iteration $\mathbf{i+H}$, the root agent is aware of the best individual in \textit{P} at least up to iteration $\mathbf{i}$.
\label{lem1}
\end{lemma}
\begin{proof}
Suppose, the best individual up to iteration $\mathbf{i}$ is found at iteration $\mathbf{i^\prime}\le \mathbf{i}$ by agent $\mathbf{a_x}$ at level $\mathbf{l^\prime}$. Afterwards, one of the following 2 cases will occur at each iteration.
\begin{itemize}
    \item Case 1. This individual will be reported to the parent of the current agent through a Found message.
    \item Case 2. This individual gets replaced by a better individual on its way to the root at iteration $\mathbf{i^{*}}>\mathbf{i^\prime}$ by agent $\mathbf{a_{y}}$ at level\:$\mathbf{l^{*}.}$
\end{itemize}
When only Case 1 occurs, the individual will reach the root at iteration $\mathbf{i^{\prime}}+\mathbf{l^{\prime}} \le \mathbf{i}+\mathbf{H}$ (since ${l^{\prime}}$ can be at most H). If Case 2 occurs, the replaced individual will reach the root agent by $\mathbf{i^*}+\mathbf{l^{*}} = \{\mathbf{i^*} - (\mathbf{l^{\prime}}-\mathbf{l^{*}}) \}+ \{(\mathbf{l^{\prime}}-\mathbf{l^{*}}) + \mathbf{l^{*}}\} = \mathbf{i^{\prime}}+\mathbf{l^{\prime}} \le \mathbf{i}+\mathbf{H}$. The same can be shown when the new individual also gets replaced. In either case, at iteration \textbf{i+H}, the root will become aware of the best individual in \textit{P} up to iteration \textbf{i} or will become aware of a better individual in \textit{P} found at iteration $\mathbf{i^{*}}>\mathbf{i}$; meaning the root will be aware of the best individual in \textit{P} at least up to iteration \textbf{i}. 
\end{proof}
\begin{lemma}
The variable assignment decision made by all the agents at iteration $\mathbf{i+2H-1}$ yield a global cost equal to the fitness of the best individual in \textit{P} at least up to iteration $\mathbf{i}$.
\label{lem2}
\end{lemma}
\begin{proof}
At iteration $\mathbf{i+2H-1}$, all the agents make decisions about variable assignment using $\mathbf{GB^{i+H}}$. However, $\mathbf{GB^{i+H}}$ is the best individual known to the root up to iteration $\mathbf{i+H}$. We know from Lemma~\ref{lem1} that, at iteration $\mathbf{i+H}$, the root is aware of the best individual in \textit{P} at least up to iteration $\mathbf{i}$. Hence, the fitness of $\mathbf{GB^{i+H}}$ is at least equal to the best individual in \textit{P} up to iteration $\mathbf{i}$. Hence, at iteration $\mathbf{i+2H-1}$, it yields a global cost equal to the fitness of the best individual in \textit{P} at least up to iteration $\mathbf{i}$.
\end{proof}
\begin{proposition}
AED is anytime.
\end{proposition}
\begin{proof}
From Lemma~\ref{lem2}, the decisions regarding the variable assignments at iterations $\mathbf{i+2H-1}$ and $\mathbf{i+2H-1+\delta}$ yields a global cost equal to the fitness of the best individual in \textit{P} at least up to iterations $\mathbf{i}$ and  $\mathbf{i+\delta}$ ($\mathbf{\delta \ge 0}$), respectively. Now, the fitness of the best individual in \textit{P} up to iteration $\mathbf{i+\delta}$ is at most the fitness at iteration $\mathbf{i}$. So the global cost at iteration $\mathbf{i+\delta}$ is less than or equal to the same cost at iteration $\mathbf{i}$. As a consequence, the quality of the solution monotonically improves as the number of iterations increases. Hence, AED is anytime.   
\end{proof}
We now consider algorithm complexity. Assume, \textbf{n} is the number of agents, $\mathbf{|N|}$ is the number of neighbours and $\mathbf{|D|}$ is the domain size of an agent. In every iteration, an agent sends $2|N|$ messages during the Reproduction step. Additionally, at most $|N|$ messages are passed for each of the ANYTIME-UPDATE and Migration steps. Now, $|N|$ can be at most n (complete graph). Hence, the total number of messages transmitted per agent during an iteration is $O(4|N|) = O(n)$. Since the main component of a message in AED is the set of individuals, the size of a single message can be calculated as the size of an individual multiplied by the number of individuals. During the Reproduction, Migration and ANYTIME-UPDATE steps, at most ER individuals, each of which has size $O(n)$, is sent in a single message. As a result, the size of a single message is $O(ER*n)$. This makes the total message size per agent during an iteration $O(ER*n*n) = O(n^2)$.\par
Before Reproduction, $|P_{a_i}|$ can be at most $2ER*|N|$ (if Migration occurred in the previous iteration) and Reproduction will add $ER*|N|$ individuals. So the memory requirement per agent is $O(3*ER*|N|*n) = O(n^2)$. Finally, Reproduction using Equations ~\ref{eqw1},~\ref{eqw2},~\ref{eqw3},~\ref{eqg1} and ~\ref{eqdel} requires $|D_i|*|N|$ operations and in total $ER*|N|$ individuals are reproduced during an iteration per agent. Hence, the total computation complexity per agent during an iteration is $O(ER*|N|*|D|*|N|) = O(|D|*n^2)$.\par 

\section{Experimental Results}
In this section, we empirically evaluate the quality of solutions produced by AED compared to six different state-of-the-art DCOP algorithms. We show that AED asymptotically converges to solutions of quality higher than these six state-of-the-art algorithms. We select these algorithms to represent all four classes of non-exact algorithms. Firstly, among the local search algorithms, we pick DSA (type C, P = 0.8, this value of P yielded the best performance in our settings), MGM2 (with offer probability p = 0.5) and GDBA (N, NM, T; reported to perform the best \cite{okamoto2016distributed}). Secondly, among the inference-based non-exact algorithms, we compare with Max-Sum\_ADVP$-$ as it has empirically shown to perform significantly better than Max-Sum \cite{zivan2012max}. We used switching parameter that yielded best result between n and 2n, where n is the number of agents. Thirdly, we consider a sampling-based algorithm, namely PD-Gibbs, which is the only such algorithm that is suitable for large-scale DCOPs \cite{dgibbs}. Finally, we compare with ACO\_DCOP as it is only available population-based DCOP algorithm. To evaluate ACO\_DCOP, we use the same values of  the parameters recommended in \cite{Chen2018AnAA}. We discuss parameter settings of AED\footnote{For implementing $Select_{wrp}(.)$ we use Reservoir-sampling algorithm \cite{swrp}. For performing set operation we use constant time polynomial hashing.} in details later in this section. Additionally, we used the ALS framework for non-monotonic algorithms having no anytime update mechanism.\par 
We compare these algorithms on three different benchmarks. We consider random DCOPs for our first benchmark. Specifically, we set the number of agents to 70 and domain size to 10. We use Erd{\H{o}}s-R{\'e}nyi topology (i.e. random graph) to generate the constraint graphs with the value of $p = 0.1$ (i.e. sparse graph) \cite{erdHos1960evolution}. We then take constraint costs uniformly from the range $[1, 100]$. Our second benchmark is identical to the first setting except the value of $p = 0.6$ (i.e. dense graph). For our last benchmark, we consider weighted graph coloring problems with the number of agents 120, 3 colors per agent, Erd{\H{o}}s-R{\'e}nyi topology with $p = 0.05$ and constraint violation costs are selected uniformly from $[1,100]$. In all three settings, we run all algorithms on 70 independently generated problems and 30 times on each problem. Moreover, for stopping condition we consider both max-iteration and max-time. For max-iteration, we stop each of the algorithms after the 1000-th iteration. For max-time, we run each algorithm for 4 seconds, 25 seconds and 6 seconds for the aforementioned benchmarks 1, 2 and 3, respectively. In order to conduct these experiments, we use a \textit{GCP-n2-highcpu-64 instance}\footnote{64 Intel Skylake vCPU @ 2.0 GHZ and 58 GB RAM} - a cloud computing service which is publicly accessible at cloud.google.com. It is worth noting that all differences shown in Figures \ref{R10}, \ref{R60}, \ref{W05} and Table \ref{time} are statistically significant for $p-value<0.01$. \par
\begin{figure}[t]
\centering
  \includegraphics[scale = 0.75]{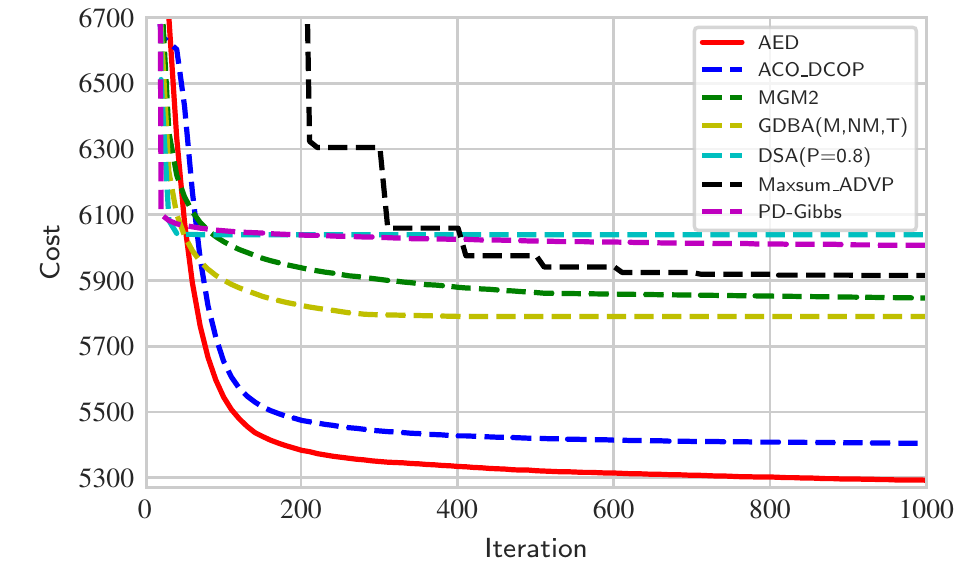}
  \vspace{-3mm}
  \caption{Comparison of AED and the benchmarking algorithms on a sparse configurations of random DCOPs.}
  \label{R10}
  \vspace{-2.5mm}
\end{figure}
Figure \ref{R10} shows a comparison between AED and the benchmarking algorithms on the sparse random DCOP benchmark after running an equal amount of iteration. On the other hand, the EXP-1 column of Table \ref{time} shows the comparison after running each algorithm an equal amount of time (4 seconds). The closest competitor to AED is ACO\_DCOP. Unlike other competitors, both of the population-based algorithms kept on improving the solution until the end of the run due to their superior capability of exploration. However, it can be observed from Table \ref{time} that AED produces 1.7\% better solution than ACO\_DCOP after running an equal amount of time. 
In contrast, most of the local search algorithms converge to local optima within 400 iterations (see Figure \ref{R10}) - with GDBA producing the best performance. After running an equal amount of time, AED outperforms GDBA by a $9\%$ and DSA by $14.9\%$. Finally, the other two representative algorithms, Max-Sum\_ADVP and PD-Gibbs are outperformed by $11.1\%-13.8\%$ margin. The superiority of AED in this experiment indicates that the Selection method along with the new Reproduction mechanism based on optimistic local benefit achieves a better balance between exploration and exploitation. This helps AED to explore until the end of the run and produce solutions with better quality than the state-of-the-art algorithms.\par 
\begin{figure}[t]
\centering
  \includegraphics[scale = 0.75]{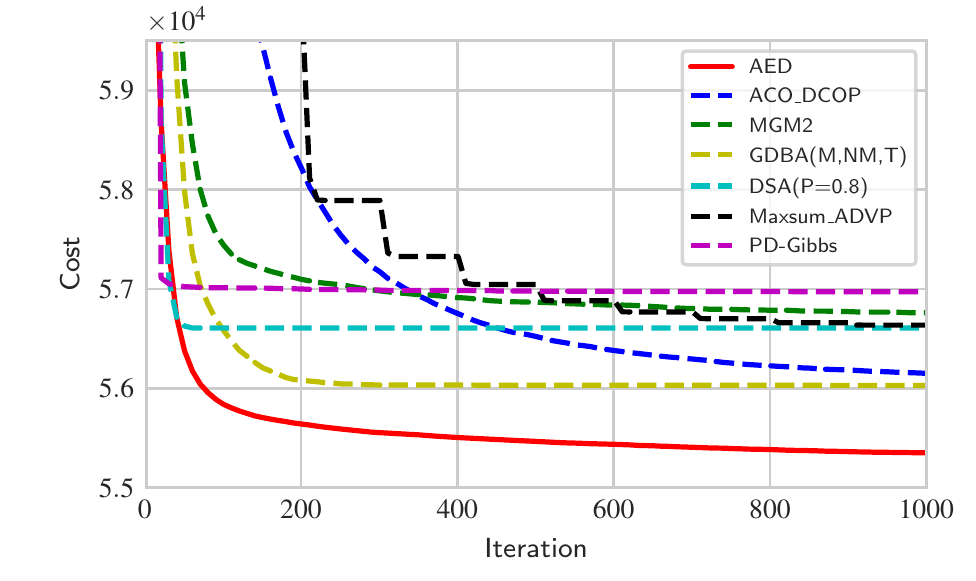}
  \vspace{-3mm}
  \caption{: Comparison of AED and the benchmarking algorithms on a dense configurations of random DCOPs.}
  \label{R60}
\end{figure}
\begin{figure}[t]
\centering
  \includegraphics[scale = 0.75]{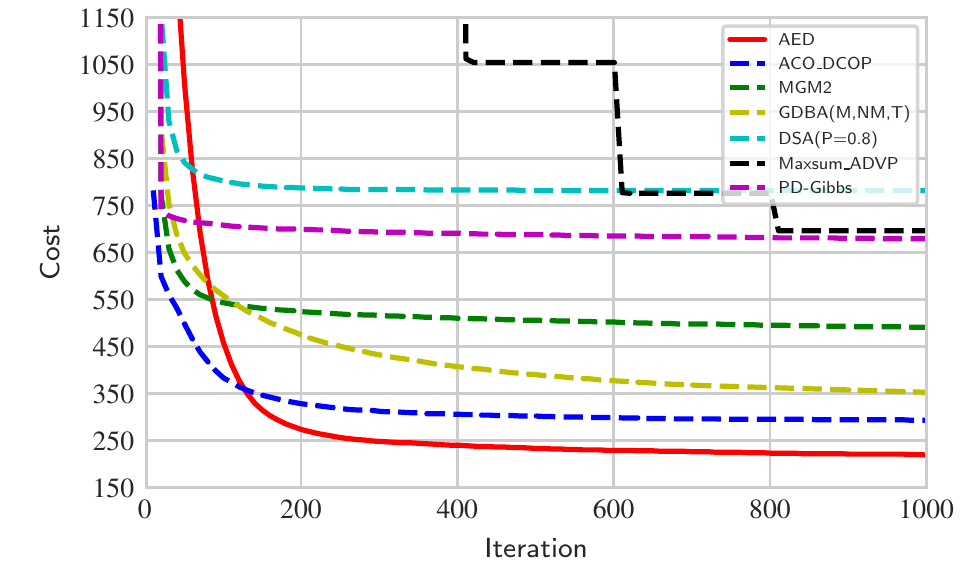}
  \vspace{-3mm}
  \caption{Comparison of AED and the benchmarking algorithms on weighted graph coloring problems.}
  \label{W05}
\end{figure}

Figure \ref{R60} shows a comparison between AED and other benchmarking algorithms on dense random DCOP benchmark. It clearly shows the advantage of AED over its competitors. To be exact, it outperforms the benchmarking algorithms by a margin of $0.7\%-3.0\%$ after running an equal amount of time (25 seconds). In this benchmatk, most of the algorithms find results of similar quality with a slight variation. Among the competitors, GDBA outperforms ACO\_DCOP by a slight margin till 1000-th iteration. However, after running an equal amount of time, ACO\_DCOP manages to produce better solutions and becomes the closest competitor to AED. PD-Gibbs fails to explore much through sampling and converges quickly while producing the most substantial performance difference with AED. It is also worth noting that  ACO\_DCOP takes 1000 iterations to produce a similar quality solution that is found by AED at the expense of only 70 iterations.\par

\begin{table}[!h]
\centering
\setlength{\tabcolsep}{2.5pt}
\caption{Comparison of AED and the benchmarking algorithms using Max-Time as Stopping condition.}
\label{time}
\begin{tabular}{|l|l|l|l|}
\hline
\textbf{Algorithm} & \textbf{EXP - 1}       & \textbf{EXP - 2}        & \textbf{EXP - 3}      \\ \hline
DSA       & 6076          & 56799          & 781          \\ \hline
MGM-2     & 5775          & 56780          & 486          \\ \hline
GDBA      & 5770          & 56051          & 310          \\ \hline
PD-Gibbs  & 6021          & 56985          & 682          \\ \hline
MS\_ADVP  & 5877          & 56786          & 625          \\ \hline
ACO\_DCOP & 5380          & 55735          & 291          \\ \hline
\textbf{AED}       & \textbf{5289} & \textbf{55347} & \textbf{229} \\ \hline
\end{tabular}
\end{table}
Figure~\ref{W05} shows a comparison between AED and the other benchmarking algorithms on weighted graph colouring problems. In this experiment, AED demonstrates its excellent performance by outperforming other algorithms by a significant margin. Among the benchmarking algorithms, ACO\_DCOP is the closest but still outperformed by AED by a 27\% margin. Among the local search algorithms, GDBA is the most competitive, but AED still finds solutions that are 35\% better. Finally, it improves the quality of solutions around $1.73-2.4$ times over some of its competitors, namely DSA, Max-Sum\_ADVP and PD-Gibbs after running an equal amount of time. Through this experiment, it is also evident that AED can also be an effective algorithm for DCSPs\par
\begin{figure}[!h]
\centering
  \includegraphics[scale = 0.75]{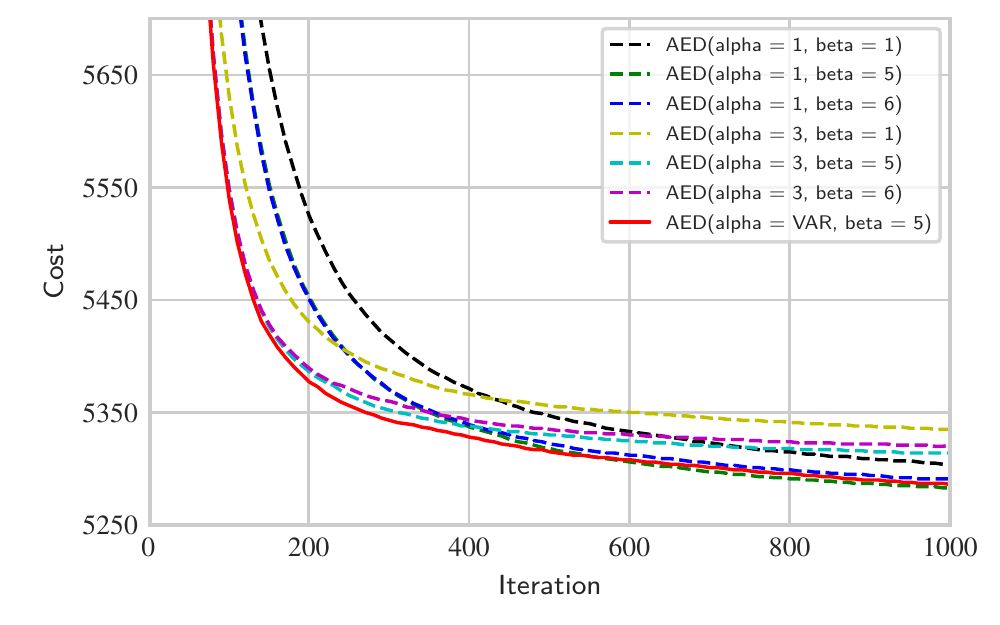}
  \vspace{-3mm}
  \caption{Performance of AED for different $\alpha$ and $\beta$ on a sparse configurations of random DCOPs.}
  \label{abd}
\end{figure}
Now we consider the effects of different parameters on the benchmarks. Firstly, for all three benchmarks, we set $IN$, which defines the initial population size to 50. A small value of $IN$ will affect the exploration of AED. However, after 50, it does not have any significant effect on the solution quality. Secondly, for all three benchmarks, we set the migration interval, $MI$ to 5. Through the Migration process, individuals get to traverse the network and different agents get to change different variables of an individual. Hence, Migration works as an implicit cooperation mechanism. If the value of $MI$ is set too high, convergence will slow down due to a lack of cooperation. On the other hand, when it is set too low, the population will lack diversity as different sub-population will mix fast. Thirdly, we show the effect of parameter $ER$ on solution quality and memory requirement on Tables \ref{ERR}. $ER$ effectively determines the population size. When it is set too low, exploration will suffer. However, as we increase $ER$ after a certain threshold, it does not improve solution quality by any significant margin. We specifically highlight the different values of ER we used in different benchmarks on Table \ref{ERR}. Notice that even with the small value of $ER=5$ AED is outperforming the benchmarking algorithms.\par        
\begin{table}[t]
	\centering
	\setlength{\tabcolsep}{2.5pt}
	\caption{Solution Quality \& Memory Requirement Per-Agent of AED for different ER value.}
	\label{ERR}
	\begin{tabular}{|l|l|l|l|l|l|l|}
		\hline
		\multicolumn{1}{|c|}{\textbf{ER}} & \multicolumn{3}{c|}{\textbf{Solution Quality}}                                                                  & \multicolumn{3}{c|}{\textbf{Memory (KB)}}                                                             \\ \cline{2-7} 
		\multicolumn{1}{|c|}{}                             & \multicolumn{1}{c|}{\textbf{EXP-1}} & \multicolumn{1}{c|}{\textbf{EXP-2}} & \multicolumn{1}{c|}{\textbf{EXP-3}} & \multicolumn{1}{c|}{\textbf{EXP-1}} & \multicolumn{1}{c|}{\textbf{EXP-2}} & \multicolumn{1}{c|}{\textbf{EXP-3}} \\ \hline
		05                                                 & 5378                                & 55550                               & 275                                 & 39                                  & 188                                 & 53                                  \\ \hline
		10                                                 & 5346                                & 55450                               & 252                                 & 69                                  & 367                                 & 97                                  \\ \hline
		20                                                 & 5316                                & \textbf{55347}                      & 240                                 & 129                                 & \textbf{725}                        & 184                                 \\ \hline
		40                                                 & \textbf{5289}                       & 55325                               & \textbf{229}                        & \textbf{248}                        & 1441                                & \textbf{358}                        \\ \hline
		50                                                 & 5285                                & 55310                               & 223                                 & 308                                 & 1899                                & 445                                 \\ \hline
	\end{tabular}
\end{table}
Finally, we depict the effect of $\alpha$ and $\beta$ in Figure \ref{abd}. While keeping $\alpha$ constant as we increase $\beta$, both the solution quality and the convergence rate increase up to a threshold. After that, the convergence rate does not change much but the solution quality starts to suffer. As we increase $\beta$, the Reproduction mechanism starts to exploit more than explore. At the threshold value, the balance between exploitation and exploration becomes optimal. After that, when we increase $\beta$, the exploration starts to suffer. Hence, this phenomenon occurs. In Figure \ref{abd}, we observe that this $\beta$ threshold is 5 (Benchmark 1). For Benchmark 2, we have found this threshold to be 5 and the value is 2 for Benchmark 3. On the other hand, as we increase $\alpha$, the convergence rate increases but the solution quality decreases. In order to mitigate this problem, we use a variable $\alpha$. To be precise, in the first 150 iterations, we use $\alpha = 3$. We then use $\alpha = 2$ in the following 150 iterations, and for the rest of the iterations we consider $1$ as the value of $\alpha$. Figure~\ref{abd} shows that  $alpha=VAR$ yields a similar solution quality as $\alpha = 1$; however, the convergence rate is near $\alpha = 3$.

\section{Conclusions}
In this paper, we introduce a novel algorithm called AED that effectively uses evolutionary optimization to solve DCOPs. To incorporate the anytime property in AED, we also present a new anytime update mechanism. In our theoretical evaluation, we prove that AED is anytime. Finally, we present empirical results that show that AED outperforms state-of-the-art non-exact algorithms by $1.7\%-14.9\%$ on sparse random DCOPs, $0.7\%-3.0\%$ on dense random DCOPs. More notably, AED produces $0.27-2.4$ times better solutions on weighted graph colouring problems. These results demonstrate the significance of applying evolutionary optimization techniques in solving DCOPs. In the future, we intend to investigate whether this algorithm can be applied to solve continuous-valued and multi-objective DCOPs. 
\section{Acknowledgments}
This research is partially supported by the ICT Division of Bangladesh Government and University Grants Commission of Bangladesh.

\clearpage
\bibliographystyle{ACM-Reference-Format}  
\bibliography{sample-bibliography}  


\begin{thebibliography}{00}


\ifx \showCODEN    \undefined \def \showCODEN     #1{\unskip}     \fi
\ifx \showDOI      \undefined \def \showDOI       #1{#1}\fi
\ifx \showISBNx    \undefined \def \showISBNx     #1{\unskip}     \fi
\ifx \showISBNxiii \undefined \def \showISBNxiii  #1{\unskip}     \fi
\ifx \showISSN     \undefined \def \showISSN      #1{\unskip}     \fi
\ifx \showLCCN     \undefined \def \showLCCN      #1{\unskip}     \fi
\ifx \shownote     \undefined \def \shownote      #1{#1}          \fi
\ifx \showarticletitle \undefined \def \showarticletitle #1{#1}   \fi
\ifx \showURL      \undefined \def \showURL       {\relax}        \fi
\providecommand\bibfield[2]{#2}
\providecommand\bibinfo[2]{#2}
\providecommand\natexlab[1]{#1}
\providecommand\showeprint[2][]{arXiv:#2}

\bibitem[\protect\citeauthoryear{Chen, He, and He}{Chen et~al\mbox{.}}{2017}]%
        {hen2017improved}
\bibfield{author}{\bibinfo{person}{Ziyu Chen}, \bibinfo{person}{Zhen He}, {and}
  \bibinfo{person}{Chen He}.} \bibinfo{year}{2017}\natexlab{}.
\newblock \showarticletitle{An improved DPOP algorithm based on breadth first
  search pseudo-tree for distributed constraint optimization}.
\newblock \bibinfo{journal}{{\em Applied Intelligence\/}}  \bibinfo{volume}{47}
  (\bibinfo{year}{2017}), \bibinfo{pages}{607--623}.
\newblock


\bibitem[\protect\citeauthoryear{Chen, Wu, Deng, and Zhang}{Chen
  et~al\mbox{.}}{2018}]%
        {Chen2018AnAA}
\bibfield{author}{\bibinfo{person}{Ziyu Chen}, \bibinfo{person}{Tengfei Wu},
  \bibinfo{person}{Yanchen Deng}, {and} \bibinfo{person}{Cheng Zhang}.}
  \bibinfo{year}{2018}\natexlab{}.
\newblock \showarticletitle{An Ant-Based Algorithm to Solve Distributed
  Constraint Optimization Problems}. In \bibinfo{booktitle}{{\em Proceedings of
  the 32nd AAAI Conference on Artificial Intelligence}}.
\newblock


\bibitem[\protect\citeauthoryear{Choudhury, Mahmud, and Khan}{Choudhury
  et~al\mbox{.}}{2019}]%
        {Choudhury2019APS}
\bibfield{author}{\bibinfo{person}{Moumita Choudhury},
  \bibinfo{person}{Saaduddin Mahmud}, {and} \bibinfo{person}{Md.~Mosaddek
  Khan}.} \bibinfo{year}{2019}\natexlab{}.
\newblock \showarticletitle{A Particle Swarm Based Algorithm for Functional
  Distributed Constraint Optimization Problems}.
\newblock \bibinfo{journal}{{\em ArXiv\/}}  \bibinfo{volume}{abs/1909.06168}
  (\bibinfo{year}{2019}).
\newblock


\bibitem[\protect\citeauthoryear{Dorigo, Birattari, and St{\"u}tzle}{Dorigo
  et~al\mbox{.}}{2006}]%
        {Dorigo2006AntCO}
\bibfield{author}{\bibinfo{person}{Marco Dorigo}, \bibinfo{person}{Mauro
  Birattari}, {and} \bibinfo{person}{Thomas St{\"u}tzle}.}
  \bibinfo{year}{2006}\natexlab{}.
\newblock \showarticletitle{Ant colony optimization: artificial ants as a
  computational intelligence technique}.
\newblock


\bibitem[\protect\citeauthoryear{Erd{\H{o}}s and R{\'e}nyi}{Erd{\H{o}}s and
  R{\'e}nyi}{1960}]%
        {erdHos1960evolution}
\bibfield{author}{\bibinfo{person}{Paul Erd{\H{o}}s} {and}
  \bibinfo{person}{Alfr{\'e}d R{\'e}nyi}.} \bibinfo{year}{1960}\natexlab{}.
\newblock \showarticletitle{On the evolution of random graphs}.
\newblock \bibinfo{journal}{{\em Institute of Mathematics, Hungarian Academy of
  Sciences\/}}  \bibinfo{volume}{5} (\bibinfo{year}{1960}),
  \bibinfo{pages}{17--60}.
\newblock


\bibitem[\protect\citeauthoryear{Farinelli, Rogers, and Jennings}{Farinelli
  et~al\mbox{.}}{2014}]%
        {farinelli2014agent}
\bibfield{author}{\bibinfo{person}{Alessandro Farinelli}, \bibinfo{person}{Alex
  Rogers}, {and} \bibinfo{person}{Nick~R Jennings}.}
  \bibinfo{year}{2014}\natexlab{}.
\newblock \showarticletitle{Agent-based decentralised coordination for sensor
  networks using the max-sum algorithm}.
\newblock \bibinfo{journal}{{\em Autonomous agents and multi-agent systems\/}}
  \bibinfo{volume}{28} (\bibinfo{year}{2014}), \bibinfo{pages}{337--380}.
\newblock


\bibitem[\protect\citeauthoryear{Farinelli, Rogers, Petcu, and
  Jennings}{Farinelli et~al\mbox{.}}{2008}]%
        {farinelli2008decentralised}
\bibfield{author}{\bibinfo{person}{Alessandro Farinelli}, \bibinfo{person}{Alex
  Rogers}, \bibinfo{person}{Adrian Petcu}, {and} \bibinfo{person}{Nicholas~R.
  Jennings}.} \bibinfo{year}{2008}\natexlab{}.
\newblock \showarticletitle{Decentralised coordination of low-power embedded
  devices using the max-sum algorithm}. In \bibinfo{booktitle}{{\em Proceedings
  of the 7th International Conference on Autonomous Agents and Multiagent
  Systems}}.
\newblock


\bibitem[\protect\citeauthoryear{Fioretto, Yeoh, Pontelli, Ma, and
  Ranade}{Fioretto et~al\mbox{.}}{2017}]%
        {fioretto2017distributed}
\bibfield{author}{\bibinfo{person}{Ferdinando Fioretto},
  \bibinfo{person}{William Yeoh}, \bibinfo{person}{Enrico Pontelli},
  \bibinfo{person}{Ye Ma}, {and} \bibinfo{person}{Satishkumar~J. Ranade}.}
  \bibinfo{year}{2017}\natexlab{}.
\newblock \showarticletitle{A Distributed Constraint Optimization (DCOP)
  Approach to the Economic Dispatch with Demand Response}. In
  \bibinfo{booktitle}{{\em Proceedings of the 16th International Conference on
  Autonomous Agents and Multiagent Systems}}.
\newblock


\bibitem[\protect\citeauthoryear{Fogel}{Fogel}{1966}]%
        {ep}
\bibfield{author}{\bibinfo{person}{David~B. Fogel}.}
  \bibinfo{year}{1966}\natexlab{}.
\newblock \showarticletitle{Artificial Intelligence through Simulated
  Evolution}.
\newblock


\bibitem[\protect\citeauthoryear{Fogel}{Fogel}{1988}]%
        {Fogel1988AnEA}
\bibfield{author}{\bibinfo{person}{David~B. Fogel}.}
  \bibinfo{year}{1988}\natexlab{}.
\newblock \showarticletitle{An evolutionary approach to the traveling salesman
  problem}.
\newblock \bibinfo{journal}{{\em Biological Cybernetics\/}}
  \bibinfo{volume}{60} (\bibinfo{year}{1988}), \bibinfo{pages}{139--144}.
\newblock


\bibitem[\protect\citeauthoryear{Gershman, Meisels, and Zivan}{Gershman
  et~al\mbox{.}}{2009}]%
        {AMIRGER2010AsynchronousFB}
\bibfield{author}{\bibinfo{person}{Amir Gershman}, \bibinfo{person}{Amnon
  Meisels}, {and} \bibinfo{person}{Roie Zivan}.}
  \bibinfo{year}{2009}\natexlab{}.
\newblock \showarticletitle{Asynchronous Forward Bounding for Distributed
  COPs}.
\newblock \bibinfo{journal}{{\em Journal of Artificial Intelligence
  Research\/}}  \bibinfo{volume}{34} (\bibinfo{year}{2009}),
  \bibinfo{pages}{61–88}.
\newblock


\bibitem[\protect\citeauthoryear{Hirayama and Yokoo}{Hirayama and
  Yokoo}{1997}]%
        {Hirayama1997DistributedPC}
\bibfield{author}{\bibinfo{person}{Katsutoshi Hirayama} {and}
  \bibinfo{person}{Makoto Yokoo}.} \bibinfo{year}{1997}\natexlab{}.
\newblock \showarticletitle{Distributed partial constraint satisfaction
  problem}. In \bibinfo{booktitle}{{\em Principles and Practice of Constraint
  Programming-CP97}}. \bibinfo{pages}{222--236}.
\newblock


\bibitem[\protect\citeauthoryear{Holland et~al\mbox{.}}{Holland
  et~al\mbox{.}}{1975}]%
        {holland1992adaptation}
\bibfield{author}{\bibinfo{person}{John~Henry Holland} {et~al\mbox{.}}}
  \bibinfo{year}{1975}\natexlab{}.
\newblock \bibinfo{booktitle}{{\em Adaptation in natural and artificial
  systems: an introductory analysis with applications to biology, control, and
  artificial intelligence}}.
\newblock \bibinfo{publisher}{MIT press}.
\newblock


\bibitem[\protect\citeauthoryear{Khan, Tran-Thanh, and Jennings}{Khan
  et~al\mbox{.}}{2018a}]%
        {gene}
\bibfield{author}{\bibinfo{person}{Md.~Mosaddek Khan}, \bibinfo{person}{Long
  Tran-Thanh}, {and} \bibinfo{person}{Nicholas~R. Jennings}.}
  \bibinfo{year}{2018}\natexlab{a}.
\newblock \showarticletitle{A Generic Domain Pruning Technique for GDL-Based
  DCOP Algorithms in Cooperative Multi-Agent Systems}. In
  \bibinfo{booktitle}{{\em Proceedings of the 17th International Conference on
  Autonomous Agents and MultiAgent Systems}}. \bibinfo{pages}{1595–1603}.
\newblock


\bibitem[\protect\citeauthoryear{Khan, Tran-Thanh, Ramchurn, and Jennings}{Khan
  et~al\mbox{.}}{2018b}]%
        {speed}
\bibfield{author}{\bibinfo{person}{Md~Mosaddek Khan}, \bibinfo{person}{Long
  Tran-Thanh}, \bibinfo{person}{Sarvapali~D Ramchurn}, {and}
  \bibinfo{person}{Nicholas~R Jennings}.} \bibinfo{year}{2018}\natexlab{b}.
\newblock \showarticletitle{{Speeding Up GDL-Based Message Passing Algorithms
  for Large-Scale DCOPs}}.
\newblock \bibinfo{journal}{{\it Comput. J.}}  \bibinfo{volume}{61}
  (\bibinfo{year}{2018}), \bibinfo{pages}{1639--1666}.
\newblock


\bibitem[\protect\citeauthoryear{Khan, Tran-Thanh, Yeoh, and Jennings}{Khan
  et~al\mbox{.}}{2018c}]%
        {nodeto}
\bibfield{author}{\bibinfo{person}{Md.~Mosaddek Khan}, \bibinfo{person}{Long
  Tran-Thanh}, \bibinfo{person}{William Yeoh}, {and}
  \bibinfo{person}{Nicholas~R. Jennings}.} \bibinfo{year}{2018}\natexlab{c}.
\newblock \showarticletitle{A Near-Optimal Node-to-Agent Mapping Heuristic for
  GDL-Based DCOP Algorithms in Multi-Agent Systems}. In
  \bibinfo{booktitle}{{\em Proceedings of the 17th International Conference on
  Autonomous Agents and MultiAgent Systems}}. \bibinfo{pages}{1604–1612}.
\newblock


\bibitem[\protect\citeauthoryear{Li}{Li}{1994}]%
        {swrp}
\bibfield{author}{\bibinfo{person}{Kim-Hung Li}.}
  \bibinfo{year}{1994}\natexlab{}.
\newblock \showarticletitle{Reservoir-Sampling Algorithms of Time Complexity
  O(n(1 + Log(N/n)))}.
\newblock \bibinfo{journal}{{\it ACM Trans. Math. Software}}
  \bibinfo{volume}{20} (\bibinfo{year}{1994}), \bibinfo{pages}{481–493}.
\newblock


\bibitem[\protect\citeauthoryear{Litov and Meisels}{Litov and Meisels}{2017}]%
        {litov2017forward}
\bibfield{author}{\bibinfo{person}{Omer Litov} {and} \bibinfo{person}{Amnon
  Meisels}.} \bibinfo{year}{2017}\natexlab{}.
\newblock \showarticletitle{Forward bounding on pseudo-trees for DCOPs and
  ADCOPs}.
\newblock \bibinfo{journal}{{\em Artificial Intelligence\/}}
  \bibinfo{volume}{252} (\bibinfo{year}{2017}), \bibinfo{pages}{83--99}.
\newblock


\bibitem[\protect\citeauthoryear{Maheswaran, Pearce, and Tambe}{Maheswaran
  et~al\mbox{.}}{2004a}]%
        {Maheswaran2004Distributed}
\bibfield{author}{\bibinfo{person}{Rajiv~T Maheswaran},
  \bibinfo{person}{Jonathan~P Pearce}, {and} \bibinfo{person}{Milind Tambe}.}
  \bibinfo{year}{2004}\natexlab{a}.
\newblock \showarticletitle{Distributed Algorithms for DCOP: A
  Graphical-Game-Based Approach.}. In \bibinfo{booktitle}{{\em Proceedings of
  the ISCA PDCS}}. \bibinfo{pages}{432--439}.
\newblock


\bibitem[\protect\citeauthoryear{Maheswaran, Tambe, Bowring, Pearce, and
  Varakantham}{Maheswaran et~al\mbox{.}}{2004b}]%
        {Maheswaran2004TakingDT}
\bibfield{author}{\bibinfo{person}{Rajiv~T. Maheswaran},
  \bibinfo{person}{Milind Tambe}, \bibinfo{person}{Emma Bowring},
  \bibinfo{person}{Jonathan~P. Pearce}, {and} \bibinfo{person}{Pradeep
  Varakantham}.} \bibinfo{year}{2004}\natexlab{b}.
\newblock \showarticletitle{Taking {DCOP} to the Real World: Efficient Complete
  Solutions for Distributed Multi-Event Scheduling}. In
  \bibinfo{booktitle}{{\em Proceedings of the 3rd International Conference on
  Autonomous Agents and Multiagent Systems}}.
\newblock


\bibitem[\protect\citeauthoryear{Modi, Shen, Tambe, and Yokoo}{Modi
  et~al\mbox{.}}{2005}]%
        {modi2005adopt}
\bibfield{author}{\bibinfo{person}{Pragnesh~Jay Modi}, \bibinfo{person}{Wei-Min
  Shen}, \bibinfo{person}{Milind Tambe}, {and} \bibinfo{person}{Makoto Yokoo}.}
  \bibinfo{year}{2005}\natexlab{}.
\newblock \showarticletitle{ADOPT: Asynchronous distributed constraint
  optimization with quality guarantees}.
\newblock \bibinfo{journal}{{\em Artificial Intelligence\/}}
  \bibinfo{volume}{161} (\bibinfo{year}{2005}), \bibinfo{pages}{149--180}.
\newblock


\bibitem[\protect\citeauthoryear{Okamoto, Zivan, Nahon, et~al\mbox{.}}{Okamoto
  et~al\mbox{.}}{2016}]%
        {okamoto2016distributed}
\bibfield{author}{\bibinfo{person}{Steven Okamoto}, \bibinfo{person}{Roie
  Zivan}, \bibinfo{person}{Aviv Nahon}, {et~al\mbox{.}}}
  \bibinfo{year}{2016}\natexlab{}.
\newblock \showarticletitle{Distributed Breakout: Beyond Satisfaction.}. In
  \bibinfo{booktitle}{{\em Proceedings of the 30th International Joint
  Conference on Artificial Intelligence}}.
\newblock


\bibitem[\protect\citeauthoryear{Ottens, Dimitrakakis, and Faltings}{Ottens
  et~al\mbox{.}}{2012}]%
        {Ottens2012DUCTAU}
\bibfield{author}{\bibinfo{person}{Brammert Ottens}, \bibinfo{person}{Christos
  Dimitrakakis}, {and} \bibinfo{person}{Boi Faltings}.}
  \bibinfo{year}{2012}\natexlab{}.
\newblock \showarticletitle{DUCT: An Upper Confidence Bound Approach to
  Distributed Constraint Optimization Problems}.
\newblock \bibinfo{journal}{{\em ACM TIST\/}}  \bibinfo{volume}{8}
  (\bibinfo{year}{2012}), \bibinfo{pages}{69:1--69:27}.
\newblock


\bibitem[\protect\citeauthoryear{Petcu and Faltings}{Petcu and
  Faltings}{2005}]%
        {Petcu2005ASM}
\bibfield{author}{\bibinfo{person}{Adrian Petcu} {and} \bibinfo{person}{Boi
  Faltings}.} \bibinfo{year}{2005}\natexlab{}.
\newblock \showarticletitle{A Scalable Method for Multiagent Constraint
  Optimization}. In \bibinfo{booktitle}{{\em Proceedings of the 19th
  International Joint Conference on Artificial Intelligence}}.
\newblock


\bibitem[\protect\citeauthoryear{Thien, Yeoh, Lau, and Zivan}{Thien
  et~al\mbox{.}}{2019}]%
        {dgibbs}
\bibfield{author}{\bibinfo{person}{Nguyen Thien}, \bibinfo{person}{William
  Yeoh}, \bibinfo{person}{Hoong Lau}, {and} \bibinfo{person}{Roie Zivan}.}
  \bibinfo{year}{2019}\natexlab{}.
\newblock \showarticletitle{Distributed Gibbs: A Linear-Space Sampling-Based
  DCOP Algorithm}.
\newblock \bibinfo{journal}{{\em Journal of Artificial Intelligence
  Research\/}}  \bibinfo{volume}{64} (\bibinfo{year}{2019}),
  \bibinfo{pages}{705--748}.
\newblock


\bibitem[\protect\citeauthoryear{Tsang and Warwick}{Tsang and Warwick}{1990}]%
        {tsang1990applying}
\bibfield{author}{\bibinfo{person}{Edward~PK Tsang} {and}
  \bibinfo{person}{Terry Warwick}.} \bibinfo{year}{1990}\natexlab{}.
\newblock \showarticletitle{Applying genetic algorithms to constraint
  satisfaction optimization problems}. In \bibinfo{booktitle}{{\em Proceedings
  of the 9th European Conference on Artificial Intelligence}}.
\newblock


\bibitem[\protect\citeauthoryear{Yeoh, Felner, and Koenig}{Yeoh
  et~al\mbox{.}}{2008}]%
        {Yeoh2008BnBADOPTAA}
\bibfield{author}{\bibinfo{person}{William Yeoh}, \bibinfo{person}{Ariel
  Felner}, {and} \bibinfo{person}{Sven Koenig}.}
  \bibinfo{year}{2008}\natexlab{}.
\newblock \showarticletitle{BnB-ADOPT: An Asynchronous Branch-and-Bound DCOP
  Algorithm}.
\newblock \bibinfo{journal}{{\em Journal of Artificial Intelligence
  Research\/}}  \bibinfo{volume}{38} (\bibinfo{year}{2008}),
  \bibinfo{pages}{85--133}.
\newblock


\bibitem[\protect\citeauthoryear{Yokoo, Durfee, Ishida, and Kuwabara}{Yokoo
  et~al\mbox{.}}{1998}]%
        {yokoo1998distributed}
\bibfield{author}{\bibinfo{person}{Makoto Yokoo}, \bibinfo{person}{Edmund~H
  Durfee}, \bibinfo{person}{Toru Ishida}, {and} \bibinfo{person}{Kazuhiro
  Kuwabara}.} \bibinfo{year}{1998}\natexlab{}.
\newblock \showarticletitle{The distributed constraint satisfaction problem:
  Formalization and algorithms}.
\newblock \bibinfo{journal}{{\em IEEE Transactions on knowledge and data
  engineering\/}}  \bibinfo{volume}{10} (\bibinfo{year}{1998}),
  \bibinfo{pages}{673--685}.
\newblock


\bibitem[\protect\citeauthoryear{Zhang, Wang, Xing, and Wittenburg}{Zhang
  et~al\mbox{.}}{2005}]%
        {zhang2005distributed}
\bibfield{author}{\bibinfo{person}{Weixiong Zhang}, \bibinfo{person}{Guandong
  Wang}, \bibinfo{person}{Zhao Xing}, {and} \bibinfo{person}{Lars Wittenburg}.}
  \bibinfo{year}{2005}\natexlab{}.
\newblock \showarticletitle{Distributed stochastic search and distributed
  breakout: properties, comparison and applications to constraint optimization
  problems in sensor networks}.
\newblock \bibinfo{journal}{{\em Artificial Intelligence\/}}
  \bibinfo{volume}{161} (\bibinfo{year}{2005}), \bibinfo{pages}{55--87}.
\newblock


\bibitem[\protect\citeauthoryear{Zivan, Okamoto, and Peled}{Zivan
  et~al\mbox{.}}{2014}]%
        {zivan2014explorative}
\bibfield{author}{\bibinfo{person}{Roie Zivan}, \bibinfo{person}{Steven
  Okamoto}, {and} \bibinfo{person}{Hilla Peled}.}
  \bibinfo{year}{2014}\natexlab{}.
\newblock \showarticletitle{Explorative anytime local search for distributed
  constraint optimization}.
\newblock \bibinfo{journal}{{\em Artificial Intelligence\/}}
  \bibinfo{volume}{212} (\bibinfo{year}{2014}), \bibinfo{pages}{1--26}.
\newblock


\bibitem[\protect\citeauthoryear{Zivan and Peled}{Zivan and Peled}{2012}]%
        {zivan2012max}
\bibfield{author}{\bibinfo{person}{Roie Zivan} {and} \bibinfo{person}{Hilla
  Peled}.} \bibinfo{year}{2012}\natexlab{}.
\newblock \showarticletitle{Max/min-sum distributed constraint optimization
  through value propagation on an alternating DAG}. In \bibinfo{booktitle}{{\em
  Proceedings of the 12th International Conference on Autonomous Agents and
  Multiagent Systems}}.
\newblock


\end{thebibliography}

\end{document}